\documentclass[letterpaper, 10 pt, conference]{cssconf}  

\IEEEoverridecommandlockouts                              
\usepackage{graphicx}
\usepackage{amssymb, amsmath, amsfonts}
\usepackage{epstopdf} 
\usepackage{enumerate}

\newtheorem{theorem}{Theorem}

\newtheorem{remark}{Remark}

\title{\LARGE \bf
Detecting Integrity Attacks on Control Systems using a Moving Target Approach
}

\author{Sean Weerakkody~~~~~~~  Bruno Sinopoli
\thanks{S. Weerakkody and B. Sinopoli are with the Department of Electrical and Computer Engineering, 
			Carnegie Mellon University, Pittsburgh, PA, USA 15213. Email: {\tt\small sweerakk@andrew.cmu.edu, brunos@ece.cmu.edu}}
\thanks{S. Weerakkody is supported in part by the Department of Defense (DoD) through the National Defense Science \& Engineering Graduate Fellowship (NDSEG) Program. The work by S. Weerakkody,  and B. Sinopoli is supported by NSF grant CNS-1329936 CPS: Synergy: Collaborative Research: Event-Based Information Acquisition, Learning, and Control in High-Dimensional Cyber-Physical Systems}}

\begin{document}
\maketitle
\thispagestyle{empty}
\pagestyle{empty}
\begin{abstract}
Maintaining the security of control systems in the presence of integrity attacks is a significant challenge. In literature, several possible attacks against control systems have been formulated including replay, false data injection, and zero dynamics attacks. The detection and prevention of these attacks may require the defender to possess a particular subset of trusted communication channels. 
Alternatively, these attacks can be prevented by keeping the system model secret from the adversary. In this paper, we consider an adversary who has the ability to modify and read all sensor and actuator channels. To thwart this adversary, we introduce external states dependent on the state of the control system, with linear time-varying dynamics unknown to the adversary. We also include sensors to measure these states. The presence of unknown time-varying dynamics is leveraged to detect an adversary who simultaneously aims to identify the system and inject stealthy outputs. Potential attack strategies and bounds on the attacker's performance are provided.
\end{abstract}

\section{Introduction}
Cyber-Physical systems (CPSs), referring to the tight interconnection of sensing, communication, and control in physical spaces, are becoming widespread in today's society. Indeed, these systems will serve a significant role in several applications including transportation, water distribution, medical technologies, manufacturing, and of course the smart grid. Due to the proliferation of CPSs in critical infrastructures, their safety and security are of paramount importance. There have already been several powerful attacks against CPSs. One major example is Stuxnet, which targeted Supervisory Control and Data Acquisition (SCADA) systems at uranium enrichment facilities in Iran \cite{Chen2010,Langner2013}. Here, the adversary was able to appropriate controllers running centrifuges at the plant, and avoid detection by replaying previous measurements to the system operator. An additional example is the Maroochy Shire incident where a disgruntled employee performed an attack on a SCADA based sewage control system \cite{Slay2008}.

Previous work \cite{Cardenas:2008ke} has suggested that existing tools in cyber security are insufficient to address attacks on CPSs due to the underlying physical system. Two main classes of attacks defined by \cite{Cardenas:2008ke} are denial of service attacks where an attacker restricts the flow of information between the plant and control center, and integrity attacks where an adversary can alter control inputs and sensor outputs. An intelligent adversary can potentially  
cause physical damage to a system using access to control inputs while manipulating sensor measurements to avoid detection. As such, integrity attacks are the main focus of this paper.

Several integrity attacks have been investigated in the literature. For instance, \cite{teixeira2012, PasqualettiJournal}  analyze zero dynamics attacks where an adversary injects inputs into both the actuators and sensors so as to bias the state without inserting a net bias on the sensor measurements. False data injection attacks on measurements, where an adversary alters a subset of sensor measurements to induce destabilizing control inputs from the defender have also been studied.  Liu et. al.~\cite{liu2009} first studied false data injection attacks in the context of electricity grids. Furthermore, in \cite{moscs10security}, the authors consider false data injection in control systems, providing sufficient and necessary conditions for an attacker to destabilize a system while introducing a bounded bias on measurement residues. Finally, replay attacks where an adversary repeats a sequence of past measurements are analyzed in \cite{Chabukswar2013, Mo2014}. 

The detection and prevention of integrity attacks on control systems against adversaries who are aware of the system model rely on the presence of one or more secure communication channels between the  operator and the plant. For instance, \cite{PasqualettiJournal} provides sufficient and necessary conditions for zero dynamics attacks based on the actuators and sensors in possession of the adversary. If the adversary has access to all sensors and actuators, a trivial zero dynamics attack is to subtract ones influence from the true measurements. To prevent false data injection attacks in control systems, a particular subset of measurements must be secure from the adversary \cite{moscs10security}. Moreover, \cite{henrik2010} proposes assigning security indices to  each sensor to quantify the effort required for an adversary to introduce a successful false data injection attack. Physical watermarking, used to detect replay attacks in \cite{Chabukswar2013, Mo2014} and robust attacks defined in \cite{Weerakkody2014}, relies on the ability to inject secret noisy inputs into the control system.  Also, \cite{Fawzi2014} which considers the problem of robust estimation and control in the presence of integrity attacks, relies on the assumption that the attacker is only able to manipulate less than half the sensors.

In this paper, we consider the scenario where an adversary has access to all communication channels. Thus, to prevent an attack, an adversary must not be aware of the full system model. \cite{teixeira2012pt2} considers the problem of altering system matrices to avoid zero dynamics attacks. However, in practice an adversary can use his access to both inputs and outputs to identify the system. Moreover, a malicious insider such as the attacker in the Maroochy Shire incident might be aware of the system model. Consequently, we propose introducing extraneous states correlated to the ordinary states of the system so that modification of the original states will impact the extraneous states. The extraneous states will have linear time-varying dynamics, known to the system operator and hidden from the adversary. The dynamics act as a moving target, changing fast enough so the adversary does not have adequate opportunity to identify the extraneous system. In this scenario, we propose attacks for the adversary and obtain detection bounds.

The rest of the paper is organized as follows. In Section II, we introduce our system model and control strategy. In Section III, we propose the moving target approach to detect integrity attacks on control systems. In Section IV, we summarize the attacker's capabilities and propose two attack models. In Section V, we analyze bounds on the attacker's performance. Section VI concludes the paper.

\section{System Model}
In this section, we introduce the model for our system. In particular, we assume our cyber-physical system can be modeled as a discrete time control system where
\begin{align}
x_{k+1} &= Ax_k + Bu_k + w_k, \label{eq:dynamics} \\
y_k &= Cx_k + v_k.  \label{eq:sensor}
\end{align}
Here $x_k \in \mathbb{R}^n$ is the state vector at time $k$ and $u_k \in \mathbb{R}^p$ is a collection of control inputs. A suite of sensors are used to monitor the state. Here $y_k \in \mathbb{R}^m$ is a vector of sensor measurements taken at time $k$. $w_k$ is the independent and identically distributed (IID) process noise with probability distribution given by $\mathcal{N}(0,Q)$ where $Q \succ 0$. Meanwhile, $v_k$ is the IID measurement noise with distribution given by $v_k \sim \mathcal{N}(0,R)$ where $R \succ 0$. We assume that $(A,C)$ is detectable. Additionally, $(A,B)$ and $(A,Q^{\frac{1}{2}})$ are assumed to be stabilizable. 

The set of measurements $y_k$ are sent to the SCADA center in order to compute the optimal control input. For our purposes, we assume that the operator wishes to minimize a quadratic function of the states and inputs as follows
\begin{equation}
J = \lim_{T \rightarrow \infty} \frac{1}{T+1} \mathbb{E} \left[ \sum_{k = 0}^T x_k^TWx_k + u_k^TUu_k \right], \label{eq:costfunction}
\end{equation}
where $W \in \mathbb{R}^{n\times n}, ~U \in \mathbb{R}^{p \times p}$ are positive definite matrices defining the relative cost of each state and input. The optimal control input for the given cost function is a combination of a Kalman filter and a linear state feedback controller \cite{Kumar1986}.

The Kalman filter computes the minimum mean squared error state estimate $\hat{x}_{k|k}^r$ \footnote{The superscript $r$ is used to distinguish the ordinary state estimate from the state estimate obtained through the moving target model.} given the previous set of measurements up to $y_k$ denoted by $y_{1:k}$. We assume that the system has been running for a long time so that the Kalman filter has converged to a fixed gain linear estimator.
\begin{align}
&\hat{x}_{k+1|k}^r = A\hat{x}_{k|k}^r + Bu_{k}, \\
& \hat{x}_{k|k}^r = (I-KC)\hat{x}_{k|k-1}^r + Ky_k, \label{eq:KalmanEst}   \\
&K = PC^T(CPC^T+R)^{-1},  \\
 &P= APA^T + Q - APC^T(CPC^T+R)^{-1}CPA^T.
 \end{align}

The optimal control input with respect to \eqref{eq:costfunction} is given by 
\begin{equation}
u_k^* = L\hat{x}_{k|k}^r, ~~~L = -(B^TSB+U)^{-1}B^TSA, \label{eq:OptInput}
\end{equation}
and $S$ satisfies the following Riccati equation
\begin{equation}
S = A^TSA + W - A^TSB(B^TSB+U)^{-1}B^TSA.
\end{equation}

A bad data detector can be utilized to determine whether a malicious attack is occurring. Typically, the bad data detector can be written as a threshold-based detector where
\begin{equation}
g_k(\mathcal{I}_k) \overset{\mathcal{H}_1}{\underset {\mathcal{H}_0}{\gtrless}} \eta_k.  \label{general detector}
\end{equation}
Here, $\mathcal{I}_k$ is the information available to the defender. The null hypothesis $\mathcal{H}_0$ is that the system is operating normally while the alternate hypothesis $\mathcal{H}_1$ is that the system is under attack. A more specific detector will be discussed later in the article. We furthermore define the probability of detection $\beta_k$ and false alarm $\alpha$ as
\begin{equation}
\beta_k = \mbox{Pr}\left(g_k\left(\mathcal{I}_k\right) > \eta_k | \mathcal{H}_1 \right),~\alpha =  \mbox{Pr}\left(g_k\left(\mathcal{I}_k\right) > \eta_k | \mathcal{H}_0 \right).
\end{equation} 
Observe that $\alpha$ is independent of $k$ since the system is stationary under $\mathcal{H}_0$. Regardless of the information available to a system operator, an attacker with knowledge of the input to output model as well as the ability to manipulate sensor measurements and control inputs, can generate undetectable attacks \cite{Smith2011}. 

For instance, an adversary can simply subtract the influence he inserts through the control inputs from the system outputs as follows
\begin{align}
x_{k+1} &= Ax_k + B(u_k^* + u_k^a) + w_k, \\
y_k &= Cx_k + v_k + s_k^a,
\end{align}
where $s_k^a$ is given by
\begin{align}
x_{k+1}^a &= Ax_k^a + Bu_k^a, \\
s_k^a &= -Cx_k^a.
\end{align}
In this case, the attack has zero net effect on the outputs and as a result $\beta_k = \alpha$.
\section{The Moving Target}
As discussed in the previous section, an adversary who is both aware of the system model and has access to all channels can generate undetectable attacks. In this work, we propose introducing linear time-varying dynamics, unknown to the adversary, but known to the defender, into the system. The defender can leverage his knowledge of the system to detect integrity attacks by the adversary. Moreover, by introducing time-varying dynamics, the defender limits the adversary's ability to identify the system using his access to measurements and inputs. The time-varying dynamics act as a moving target.
\subsection{Extended Model}
We extend the state $x_k$ to include extraneous states $\tilde x_k \in \mathbb{R}^{\tilde n}$ as follows
\begin{equation}
 \begin{bmatrix} \tilde x_{k+1} \\ x_{k+1} \end{bmatrix} =
  \mathcal{A}_k\begin{bmatrix} \tilde x_k \\ x_k \end{bmatrix}
 + \mathcal{B}_ku_k + \begin{bmatrix} \tilde w_k \\ w_k \end{bmatrix}  ,
 \end{equation} 
 where
 \begin{equation}
 \mathcal{A}_k\triangleq \begin{bmatrix} A_{1,k} & A_{2,k} \\0 & A \end{bmatrix}, ~~~ \mathcal{B}_k \triangleq \begin{bmatrix} B_k \\ B \end{bmatrix} .
 \end{equation}
Moreover, we introduce additional sensors $\tilde y_k \in \mathbb{R}^{\tilde m}$ to measure the extraneous states.
  \begin{equation}
 \begin{bmatrix} \tilde{y}_k \\ y_k \end{bmatrix} = \mathcal{C}_k \begin{bmatrix} \tilde x_k \\ x_k \end{bmatrix} + \begin{bmatrix} \tilde v_k \\ v_k \end{bmatrix},~~ \mathcal{C}_k \triangleq \begin{bmatrix} C_k & 0 \\0 & C \end{bmatrix}.
 \end{equation}
 The matrices are assumed to be IID random variables which are independent of the sensor and process noise with distribution 
 \begin{equation}
 A_{1,k}, A_{2,k}, B_k, C_{k+1} \sim
 f_{A_{1,k},A_{2,k},B_k,C_{k+1}}(A_1,A_2,B,C).
 \end{equation}
 Furthermore, we also assume that 
 \begin{equation}
\begin{bmatrix} \tilde w_k \\ w_k \end{bmatrix}   \sim \mathcal{N}\left(0,\mathcal{Q}\right), ~~
\begin{bmatrix} \tilde v_k \\ v_k \end{bmatrix} \sim \mathcal{N}\left(0, \mathcal{R} \right),
 \end{equation}
 where 
 \begin{equation}
 \mathcal{Q} = \begin{bmatrix} \tilde Q & \tilde Q_{12} \\ \tilde Q_{12}^T & Q \end{bmatrix} \succ 0,~~~
 \mathcal{R} =  \begin{bmatrix} \tilde R & \tilde R_{12} \\ \tilde R_{12}^T & R \end{bmatrix} \succ 0.
 \end{equation}
 \begin{remark}
 While we assume the structure of the system introduced above with IID matrices $ A_{1,k}, A_{2,k}, B_k, C_{k+1}$, the moving target design can still be effective in other scenarios. For instance, the dynamics need not be linear as long as the defender can accurately model the system. Moreover, the system parameters do not have to evolve at each time step, though the longer the target remains in place, the easier it is for the adversary to identify the system. In addition, the matrices $A_{1,k}$, $A_{2,k},$ or $B_k$ may be sparse, as long as there exists adequate coupling between $x_k$ and $\tilde{x}_k$.
 \end{remark}
 \begin{remark}
 The defender must be able to introduce extraneous states with time-varying dynamics correlated to the original state of the system. The extraneous states are application dependent and are to be decided by the system operator. Nonetheless, the system operator can leverage existing waste products of the system, for instance the heat dissipated by a reaction or process. The dynamics can be made time-varying by changing conditions at the plant. Alternatively, the defender can introduce dynamics into the system. For instance, the defender can introduce RLC circuits which measure the states. Time varying dynamics can be incorporated by including variable resistors or capacitors. By varying the components of the circuit according to some IID distribution at each time step, the defender can generate IID system matrices.
 \end{remark}
\begin{remark}
 In the above formulation we assume that the defender is aware of the real time system matrices although they are random. In general, this information should not be sent over the network since doing so amounts to the existence of a secure communication channel. The secure communication channel could be leveraged to detect an attack without considering a moving target approach, for instance through physical watermarking \cite{Weerakkody2014}. Alternatively, we can generate pseudo random system matrices using a pseudo random number generator (PRNG). In this case, the seed of the PRNG will be known to the defender and kept hidden from the attacker.
\end{remark}

\subsection{Estimation and Detection}
The presence of additional sensors allows us to improve our estimate of the state. In particular, we can incorporate an additional Kalman filter to estimate the state as follows.
\begin{align}
\begin{bmatrix} \hat{\tilde{x}}_{k+1|k} \\ \hat{x}_{k+1|k} \end{bmatrix} 
&= \mathcal{A}_k\left(\left(I-\mathcal{K}_kC_k\right)\begin{bmatrix} \hat{\tilde{x}}_{k|k-1} \\ \hat{x}_{k|k-1} \end{bmatrix}  + \mathcal{K}_k\begin{bmatrix} \tilde y_k \\ y_k \end{bmatrix} \right) \nonumber
\\ &+ \mathcal{B}_kL\hat{x}_{k|k}^r,  \label{KalmanBig} \\
\mathcal K_k &= \mathcal{P}_k\mathcal{C}_k^T\left(\mathcal{C}_k\mathcal{P}_k\mathcal{C}_k^T+\mathcal{R}\right)^{-1}, \\
\mathcal{P}_{k+1} &= \mathcal{A}_k\left(\mathcal{P}_k - \mathcal{K}_k\mathcal{C}_k\mathcal{P}_k \right)\mathcal{A}_k^T + \mathcal{Q}.
\end{align}
Observe that we use the state estimate $\hat{x}_{k|k}^r$ to compute the input $u_k^*$ as opposed to an estimate derived from \eqref{KalmanBig}. We assume the defender does not care about controlling $\tilde{x}_k$. In this case, adding the moving target does not change $J$.  Such a strategy also prevents the attacker from using information from the input to learn about the system model. In fact, we have the following result.
\begin{theorem}
The input $u_k^* = L\hat{x}_{k|k}^r$ is independent from the system matrices $A_{1,k-1},A_{2,k-1},B_{k-1},C_k$ for all $k$.
\end{theorem}
\begin{proof}
The input $u_k^*$ is given by
\begin{equation}
l(\hat{x}_{0|0}^r,x_0,A,B,C,K,L,w_{0} \ldots w_{k-1},v_{1} \ldots v_{k}),
\end{equation}
where $l$ is some deterministic function of variables which by assumption are independent from $A_{1,k-1},A_{2,k-1},B_{k-1},C_k$ for all $k$. The result immediately follows.
\end{proof}
A similar result can be obtained under attack where $u_k^*$ is conditionally independent of the system matrices $A_{1,k-1},A_{2,k-1},B_{k-1},C_k$ for all $k$, given the adversary's attack inputs.

We assume that a residue based detector is incorporated where the residue $z_k$ is given by
\begin{equation}
z_k \triangleq \begin{bmatrix} \tilde y_k \\ y_k \end{bmatrix} - \mathcal{C}_k\begin{bmatrix} \hat{\tilde{x}}_{k|k-1} \\ \hat{x}_{k|k-1} \end{bmatrix} \sim \mathcal{N}\left(0,\mathcal{C}_k\mathcal{P}_k\mathcal{C}_k^T + \mathcal{R} \right).
\end{equation}
We can leverage knowledge of the distribution of $z_k$ under normal operation to design a detector. In particular we consider a $\chi^2$ detector where $g_k$ in \eqref{general detector} is given by
\begin{equation}
g_k(z_k) = z_k^T(\bar{\mathcal{P}}_k)^{-1}z_k, \label{eq:gk}
\end{equation}
where $\bar{\mathcal{P}}_k = \mathcal{C}_k\mathcal{P}_k\mathcal{C}_k+\mathcal{R}$. Under normal operation $g_k$ has a $\chi^2$ distribution. In general, the window for the detector can be extended to consider past measurements. In Figure 1, we include a diagram of the moving target system operating normally.
\begin{figure}[htb]
\includegraphics[height=55 mm]{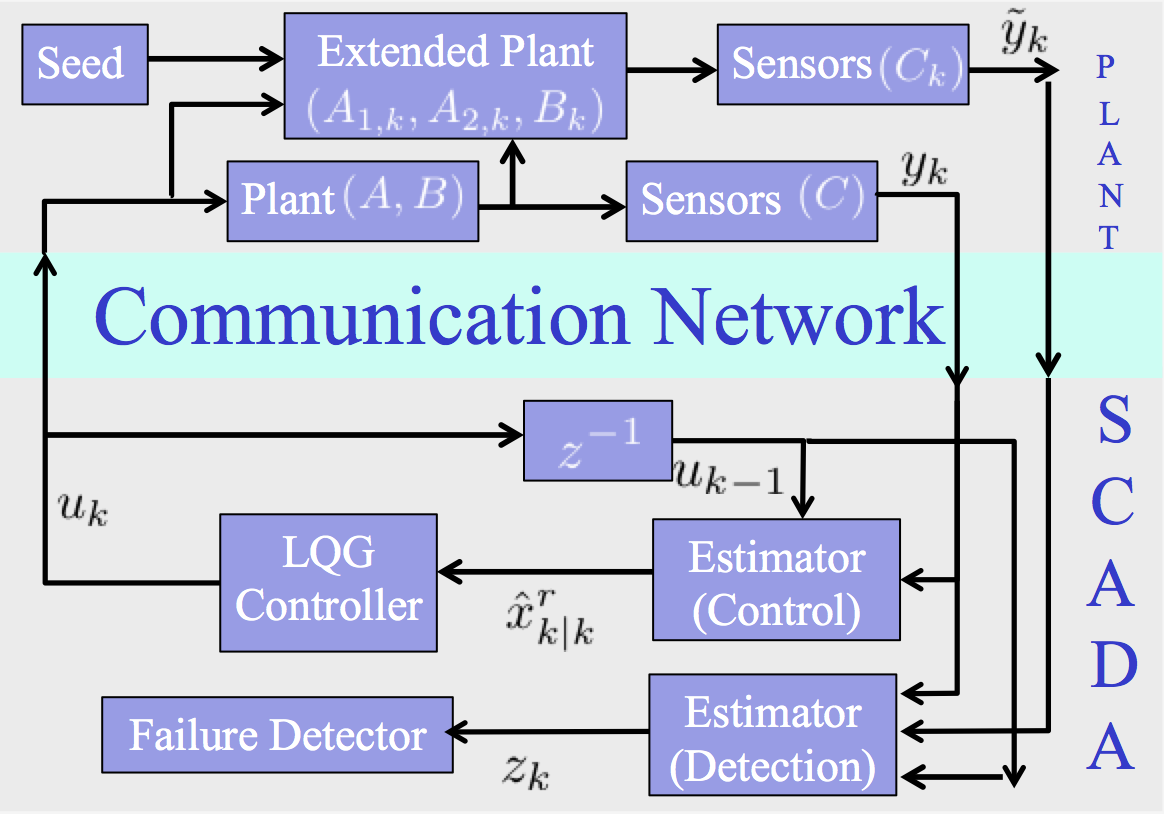}
\caption{Diagram of system under normal operation}
\end{figure} 
 \section{Attack Model}
 In this section we describe a near omnipotent attacker in terms of his capabilities, access to information, and potential strategies. On one hand, the adversary may acquire his knowledge and resources through a highly sophisticated attack strategy as done in Stuxnet.  On the other hand, an adversary can obtain his resources through insider information and access as done in the Maroochy Shire incident.
 \subsection{Attack Capabilities}
1) The attacker can insert arbitrary inputs into the system and can arbitrarily alter the sensor measurements. As a result, 
 when under attack, the system has dynamics given by 
 \begin{equation}
 \begin{bmatrix} \tilde x_{k+1} \\ x_{k+1} \end{bmatrix}  = \mathcal{A}_k\begin{bmatrix} \tilde x_k \\ x_k \end{bmatrix}
 + \mathcal{B}_k(u_k+u_k^a) + \begin{bmatrix} \tilde w_k \\ w_k \end{bmatrix}, 
 \end{equation} 
 \begin{equation}
  \begin{bmatrix} \tilde y_k^a \\  y_k^a \end{bmatrix} =  \begin{bmatrix} \tilde{y}_k \\ y_k \end{bmatrix} + \begin{bmatrix} \tilde s_k^a \\  s_k^a \end{bmatrix}.
 \end{equation}
 where $u_k^a$ is the attacker's control input and $\tilde s_k^a$ and $s_k^a$ are the biases injected on the extraneous sensors and ordinary sensors respectively.
 
2) The attacker can read the true outputs of the system $\tilde y_k, y_k$ and the inputs being sent by the defender to the plant $u_k$ for all time $k$.
 \begin{remark} The attacker essentially performs a man in the middle attack between the plant and system operator so that he can manipulate and read all communication channels arbitrarily. A malicious insider can do this by breaking encryption schemes. Furthermore, physical attacks can be used to change sensor measurements and control inputs. For instance, locally heating or cooling a temperature sensor would change the sensor measurements without violating the integrity or authenticity of data from a cyber perspective.
 \end{remark}
 
3) The attacker has full knowledge of the system model $\mathcal{S} \triangleq  \{A,B,C,K,L,\mathcal{Q},\mathcal{R}\}$. Moreover, the adversary knows the probability density function (pdf) of random matrices $A_{1,k},A_{2,k},B_k,C_{k+1}$.
 \begin{remark} While conservative, the adversary can obtain his knowledge of the system model by observing the communication channels for an extended period of time and performing system identification. Moreover, observe that since the attacker is aware of the original system model and all outputs, he can asymptotically predict the state estimate $\hat{x}_{k|k}^r$ if the matrix $(A+BL)(I-KC)$ is stable \cite{Chabukswar2013}.
 \end{remark}
 \begin{remark}
 The attacker can leverage his probabilistic knowledge of the system model as well as the true outputs of the system to generate stealthy attack inputs $s_k^a, \tilde s_k^a$. In particular, the adversary can attempt to simultaneously identify the moving target and generate convincing counterfeit sensor outputs.
 \end{remark}

Based on the above definitions we can define the private information available to the attacker and defender at time $k$ $\mathcal{I}_k^{A},\mathcal{I}_k^{D}$ and the public information $\mathcal{I}_k^{P}$ available to both as
\begin{align}
\mathcal{I}_k^{A} &\triangleq \{ \tilde y_j, y_j, \tilde s_{j-1}^a, s_{j-1}^a, u_{j-1}^a \} ~~~ \forall~ j \le k, \\
\mathcal{I}_k^{D} &\triangleq \{A_{1,j-1}, A_{2,j-1}, B_{j-1}, C_{j} \} ~~~ \forall~ j, \\
\mathcal{I}_k^{P} &\triangleq \{\mathcal{S},f(A_1,A_2,B,C), u_{j-1}, \tilde y_{j-1}^a, y_{j-1}^a\} ~~~\forall~ j \le k.
\end{align}
In Figure 2, we include a diagram of the system under attack.
\begin{figure}[htb]
\includegraphics[height=55 mm]{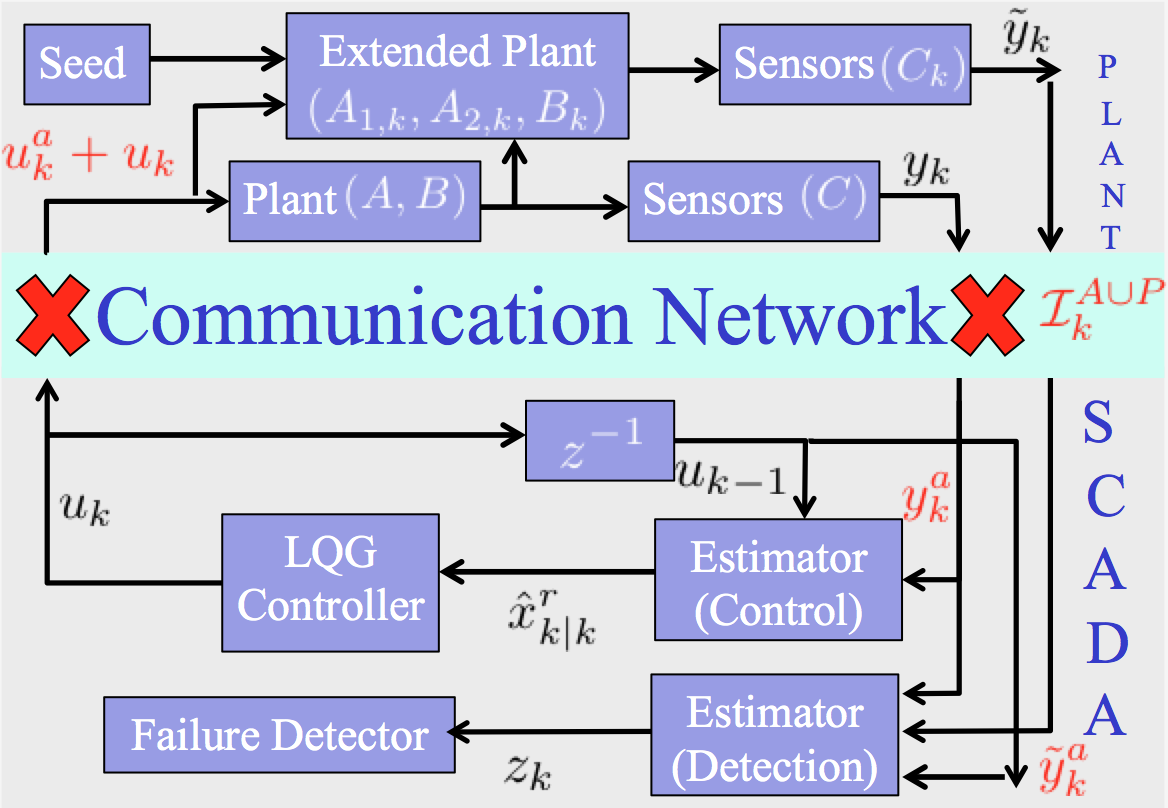}
\caption{Diagram of system under attack}
\end{figure} 
\subsection{Attack Strategy}
In this subsection we propose two main attack strategies. Without loss of generality we assume any attack begins at $k = 0$.
\subsubsection{Attack 1: Subtract Influence}
In the first attack strategy the attacker aims to estimate his influence on the control system and subtract it.
Define $ \bar{s}_k^a \triangleq [\tilde{s}_k^{a~T} ~  s_k^{a~T}]^T $. Observe that if 
\begin{equation}
\bar{x}_{k+1}^a = \mathcal{A}_k\bar{x}_k^a + \mathcal{B}_ku_k^a,~~ \Delta \bar{y}_k^a = \mathcal{C}_k\bar{x}_k^a,
\end{equation}
with initial state $\bar{x}_0^a = 0$ and $\bar{s}_k^a = -\Delta \bar{y}_k^a$, an attack is completely stealthy. As the adversary does not know the time varying matrices, we assume he computes an estimate of $\Delta \bar{y}_k^a$ and uses that to subtract his influence on the sensor measurements. Thus, we would have
\begin{equation}
\bar{s}_k^a = -\mathbb{E}[\Delta \bar y_k^a|\mathcal{I}_k^A \cup \mathcal{I}_k^P  ].
\end{equation}
\begin{remark}
Observe that the adversary can exactly subtract his influence from measurements $y_k$ due to his knowledge of the system model. However, the adversary should be unable to completely subtract his bias from the extraneous sensors $\tilde y_k$.
\end{remark}
\vspace{0.3cm}
\textit{Optimal Theoretical Estimation}
Define $\bar y_k^a \triangleq  [\tilde y_k^{aT} ~  y_k^{aT}]^T,~\bar{x}_k \triangleq [\tilde x_k^T ~  x_k^T]^T,~~\bar{w}_k \triangleq  [\tilde w_k^T ~  w_k^T]^T$, $\bar{v}_k \triangleq  [\tilde v_k^T ~  v_k^T]^T$, and $ \bar{y}_k \triangleq [\tilde y_k^T ~  y_k^T]^T$. The adversary's observations can be formulated through the following linear time-varying system,
\begin{align}
\begin{bmatrix} \bar{x}_{k+1}  \\ \bar{x}_{k+1}^a \end{bmatrix} &= \begin{bmatrix} \mathcal{A}_k & 0 \\ 0  & \mathcal{A}_k \end{bmatrix}\begin{bmatrix} \bar{x}_{k}  \\ \bar{x}_{k}^a \end{bmatrix} + \begin{bmatrix} \mathcal{B}_k  & \mathcal{B}_k \\ 0 & \mathcal{B}_k  \end{bmatrix}\begin{bmatrix} u_k \\ u_k^a \end{bmatrix} + \begin{bmatrix} \bar{w}_k \\ 0 \end{bmatrix},
\\ \bar y_k &= \begin{bmatrix} \mathcal{C}_k & 0 \end{bmatrix}\begin{bmatrix} \bar{x}_{k}  \\ \bar{x}_{k}^a \end{bmatrix}  + \bar{v}_k.
\end{align}
To estimate $\Delta \bar{y}_k^a$ at time $k$, assume the adversary has access to the following distribution $f(\bar{x}_k,\bar{x}_k^a,C_k|\mathcal{I}_k^{A \cup P})$ where $\mathcal{I}_k^{A \cup P} = \mathcal{I}_k^A \cup \mathcal{I}_k^P$  
Then we have 
\begin{equation}
\bar{s}_k^a = -\int_{\bar{x}_k}\int_{\bar{x}_k^a}\int_{C_k} \mathcal{C}_k\bar{x}_k^a f(\bar{x}_k,\bar{x}_k^a,C_k|\mathcal{I}_k^{A \cup P}) \mbox{d}\bar{x}_k\mbox{d}\bar{x}_k^a\mbox{d}C_k. \label{expectedvalues}
\end{equation}
We show that the pdf can be recursively computed at each step. Letting $\zeta_{k+1} = \{\bar{x}_{k+1},\bar{x}_{k+1}^a,C_{k+1}\}$ we have 
\begin{align}
f(\zeta_{k+1}|&\mathcal{I}_{k+1}^{A \cup P})   = f(\zeta_{k+1}|\mathcal{I}_{k}^{A \cup P} , \bar{y}_{k}^a, \bar{y}_{k+1}, \bar{s}_{k}^a, u_{k}^a,u_k), \nonumber \\
&= f(\zeta_{k+1}|\mathcal{I}_{k}^{A \cup P} , \bar{y}_{k+1}, u_{k}^a,u_k), \nonumber \\ 
&= \frac{f(\bar{y}_{k+1}| \mathcal{I}_{k}^{A \cup P},\zeta_{k+1})f(\zeta_{k+1}|\mathcal{I}_{k}^{A \cup P},u_k,u_k^a)}{ f(\bar{y}_{k+1}|\mathcal{I}_{k}^{A \cup P},u_k,u_k^a)}. \label{pdfupdates}
\end{align}
The second equality follows from the conditional independence of $\zeta_{k+1}$ and $\bar{y}_{k}^a,\bar{s}_{k}^a$ given $\bar{y}_{k}$ and $u_k$. The last equality follows from Bayes rule and the conditional independence of $\bar{y}_{k+1}$ and $u_k, u_k^a$ given $\zeta_{k+1}$. We note that this distribution can be theoretically computed given the attacker's information.
That is, we know that
\begin{equation}
f(\bar{y}_{k+1}| \mathcal{I}_{k}^{A \cup P},\zeta_{k+1}) \sim \mathcal{N}\left(\mathcal{C}_{k+1}\bar{x}_{k+1}, \mathcal{R} \right).
\end{equation}
Moreover, $\zeta_{k+1}$ and $\bar{y}_{k+1}$ are deterministic functions of $\zeta_k$, $u_k$, $u_k^a$ and random variables $A_{1,k}$, $A_{2,k}$, $B_k$, $C_{k+1}$, $\bar{w}_k$, $\bar{v}_{k+1}$ which are independent of $\zeta_k$ given $\mathcal{I}_{k}^{A \cup P}$. Thus, theoretically, $f(\zeta_{k+1}|\mathcal{I}_{k+1}^{A \cup P})$ can be recursively computed from $f(\zeta_k|\mathcal{I}_{k}^{A \cup P})$.
\begin{remark}
If the attacker subtracts his influence, he might be susceptible to a growing cancellation error if he attempts to excite the system's unstable dynamics. Instead of subtracting his influence the attacker can instead directly estimate what the defender expects to see as summarized in the next section.
\end{remark}
\subsubsection{Attack 2: Estimate Expected Measurement}
In the next strategy, the adversary aims to track the system operator's state estimate. Using the system operator's state estimate, the adversary attempts to generate stealthy outputs. Let $\hat{\bar{x}}_k = [\hat{\tilde{x}}_{k|k-1}^T \hat{x}_{k|k-1}^T]^T$. The attacker's observations and strategy can be formulated as follows
\begin{align}
\begin{bmatrix} \bar{x}_{k+1}  \\ \hat{\bar{x}}_{k+1} \end{bmatrix} &= \begin{bmatrix} \mathcal{A}_k & 0 \\ 0 & \mathcal{A}_k(I-\mathcal{K}_k\mathcal{C}_k) \end{bmatrix}\begin{bmatrix} \bar{x}_{k}  \\ \hat{\bar{x}}_{k} \end{bmatrix} + \begin{bmatrix} \bar{w}_k \\ 0 \end{bmatrix}, \nonumber \\ &+   \begin{bmatrix} \mathcal{B}_k  & \mathcal{B}_k & 0 \\ \mathcal{B}_k & 0 & \mathcal{A}_k\mathcal{K}_k \end{bmatrix}\begin{bmatrix} u_k \\ u_k^a \\ \bar{y}_k^a \end{bmatrix},
\end{align}
\begin{equation}
\bar{y}_k = \begin{bmatrix} \mathcal{C}_k & 0 \end{bmatrix}\begin{bmatrix} \bar{x}_{k}  \\ \hat{\bar{x}}_{k} \end{bmatrix} + \bar{v}_k, ~~~~\bar{s}_k^a = \mathbb{E}[\mathcal{C}_k\hat{\bar{x}}_{k}|\mathcal{I}_{k}^{A \cup P}]-\bar{y}_k.
\end{equation}
The attacker wishes to track $\zeta_{k} = \{\bar{x}_{k},\hat{\bar{x}}_{k},C_{k},\mathcal{P}_{k}\}$. The use of the preceding attack design is motivated by the ensuing result which states that the chosen attack vector minimizes a fixed quadratic function of the measurement residues.
\begin{theorem}
Let $\Sigma \succeq 0$ be a positive semidefinite matrix.
\begin{equation}
\mathbb{E}[\mathcal{C}_k\hat{\bar{x}}_{k}|\mathcal{I}_{k}^{A \cup P}]-\bar{y}_k = \arg \underset{\bar{s}_k^a}{\min} ~\mathbb{E}[z_k^T\Sigma z_k|\mathcal{I}_{k}^{A \cup P}].
\end{equation}
\begin{proof}
 Observe that
\begin{equation}
\mathbb{E}[z_k^T \Sigma z_k|\mathcal{I}_{k}^{A \cup P}] = \int_{\zeta_k} z_k^T \Sigma z_k f(\zeta_k|\mathcal{I}_{k}^{A \cup P})\mbox{d}\zeta_k.
\end{equation}
Taking the gradient with respect to $\bar{s}_k^a$ and setting the resulting expression equal to 0, we obtain
\begin{equation}
\int_{\zeta_k} \Sigma (\bar{y}_k + \bar{s}_k^a - \mathcal{C}_k\hat{\bar{x}}_k) f(\zeta_k|\mathcal{I}_{k}^{A \cup P})\mbox{d}\zeta_k = 0.
\end{equation}
Solving gives
\begin{equation}
\bar{s}_k^a = -\bar{y}_k + \int_{\zeta_k}\mathcal{C}_k\hat{\bar{x}}_k f(\zeta_k|\mathcal{I}_{k}^{A \cup P})\mbox{d}\zeta_k,
\end{equation}
and the result holds.
\end{proof}
\end{theorem}

To determine $\bar{s}_k^a$ at time $k$ assume the adversary has access to the following distribution $f(\zeta_k|\mathcal{I}_k^{A \cup P})$. As done before, the attacker can theoretically compute $\bar{s}_k^a$  by taking a conditional expectation.  Additionally, similar to \eqref{pdfupdates} we have
\begin{align}
&f(\zeta_{k+1}|\mathcal{I}_{k+1}^{A \cup P})  \nonumber \\
&= \frac{f(\bar{y}_{k+1}| \mathcal{I}_{k}^{A \cup P},\zeta_{k+1})f(\zeta_{k+1}|\mathcal{I}_{k}^{A \cup P},u_k,u_k^a,\bar{y}_k^a)}{ f(\bar{y}_{k+1}|\mathcal{I}_{k}^{A \cup P},u_k,u_k^a,\bar{y}_k^a)}. 
\end{align}

 Moreover, by similar analysis as in attack 1, we can demonstrate that $f(\zeta_{k+1}|\mathcal{I}_{k+1}^{A \cup P})$ can be recursively computed from $f(\zeta_k|\mathcal{I}_{k}^{A \cup P})$. The main difference here is that the adversary must also estimate $\mathcal{P}_k$. Note that in practice the proposed attacks are difficult to execute for an adversary since it is likely a challenge to compute the necessary distribution functions and expected values. As a result, in the next section we aim to provide bounds on the attacker's estimation performance in terms of mean square error matrices.

\section{Bounds on Attacker's Performance}
\subsection{Bounds on Attacker's State Estimation}
In this section we attempt to characterize lower bounds on the error matrices associated with the states $\zeta_k$ defined in attack strategy 1 and 2. From there, we can attempt to characterize how well the adversary can design $\bar{s}_k^a$ to fool the bad data detector.

We leverage conditional posterior Cramer-Rao lower bounds for Bayesian sequences  derived by \cite{Zuo2011}. 
The authors here make use of the Bayesian Cramer-Rao lower bound or Van Trees bound derived in \cite{VanTrees1968} which states that for observations $y$ and states $\zeta$ the mean squared error matrix is bounded by the Fisher information as follows
\begin{equation}
\mathbb{E}_{f(\zeta,y)}\left[[\hat\zeta(y)-\zeta][\hat\zeta(y)-\zeta]^T\right] \ge I^{-1},
\end{equation}
where the Fisher information matrix $I$ is given by
\begin{equation}
I = \mathbb{E}_{f(\zeta,y)} \left[-\triangle_{\zeta}^{\zeta} \mbox{ log} f(\zeta,y) \right].
\end{equation}
Note that
\begin{equation*}
\triangle_x^y g(x,y) \triangleq \triangledown_x \triangledown_y^T g(x,y).
\end{equation*}
In \cite{Zuo2011}, this result is extended to nonlinear Bayesian sequences with dynamics given by
\begin{align}
\zeta_{k+1} = F_k(\zeta_k,\omega_k), ~~ \bar y_k = G_k(\zeta_k,\bar{v}_k),
\end{align}
where $\omega_k$ and $\bar{v}_k$ are independent process and sensor noise respectively. In our case, we slightly adapt these results to account for the fact there is feedback in our system so that 
\begin{equation}
\zeta_{k+1} = F_k(\zeta_k,\bar{y}_{1:k},\omega_k),~~~ \bar y_k = G_k(\zeta_k,\bar{v}_k).
\end{equation}
The inputs $u_k$, $u_k^a$ and $\bar{s}_k^a$ are incorporated into the definition of $F_k$, while uncertainty in the model $(A_{1,k},A_{2,k},B_k,C_{k+1})$ can be incorporated in the process noise $\omega_k$. It can shown that the following posterior Cramer-Rao lower bound holds
\begin{equation}
\mathbb{E}_{f_{k+1}^c}\left[e_{0:k+1}e_{0:k+1}^T|\bar{y}_{1:k}\right] \ge I^{-1}(\zeta_{0:k+1}|\bar{y}_{1:k}), \label{PosteriorCramer}
\end{equation}
where
\begin{align}
e_{0:k+1} &\triangleq \zeta_{0:k+1} - \hat{\zeta}_{0:k+1}(\bar{y}_{k+1}|\bar{y}_{1:k}), \\
f_{k+1}^c &\triangleq f(\zeta_{0:k+1},\bar{y}_{k+1}|\bar{y}_{1:k}),\\
I(\zeta_{0:k+1}|\bar{y}_{1:k}) &\triangleq \mathbb{E}_{f_{k+1}^c}\left[-\triangle_{\zeta_{0:k+1}}^{\zeta_{0:k+1}} \mbox{log } f_{k+1}^c | \bar{y}_{1:k} \right].
\end{align}
\begin{remark} \label{fk1c}
We remark that since $F_k$ is defined by inputs $u_k$, $u_k^a$ and $\bar{s}_k^a$, $f_{k+1}^c$ is implicitly conditioned on $u_{0:k}, \bar{s}_{1:k}^a, u_{0:k}^a$. Moreover, $f_{k+1}^c$ is defined given the adversary's knowledge of $\mathcal{S},f(A_1,A_2,B,C)$. 
\end{remark}
Observe that \eqref{PosteriorCramer} gives us an expected lower bound for the error matrix associated with the entire state history $\zeta_{0:k+1}$ with knowledge of measurements $\bar{y}_{1:k}$. This expectation is taken over the state history as well the measurement $\bar{y}_{k+1}$ so that $\hat{\zeta}_{0:k+1}$ is a function of the measurement $\bar{y}_{k+1}$. Observe that unlike the traditional Cramer-Rao bound which is limited to unbiased estimators, the Bayesian Cramer-Rao bound here considers both biased and unbiased estimators $\hat{\zeta}$.

While the lower bound given here applies to the entire state history $\zeta_{0:k+1}$, in practice we care about estimating a lower bound on the current state $\zeta_{k+1}$. Nonetheless, it can be easily shown that
\begin{equation}
\mathbb{E}_{f_{k+1}^c}\left[e_{k+1}e_{k+1}^T|\bar{y}_{1:k}\right] \ge I^{-1}(\zeta_{k+1}|\bar{y}_{1:k}), 
\end{equation} 
where $I^{-1}(\zeta_{k+1}|\bar{y}_{1:k})$ is the $\mbox{dim}(\zeta_{k}) \times \mbox{dim}(\zeta_{k})$ lower right submatrix of $I^{-1}(\zeta_{0:k+1}|\bar{y}_{1:k})$. In practice, computing $I^{-1}(\zeta_{k+1}|\bar{y}_{1:k})$ from $I^{-1}(\zeta_{0:k+1}|\bar{y}_{1:k})$ is impractical since it requires computing and taking the inverse of a Fisher information matrix which grows in dimension at each time step. As a result, we would like a recursion to compute $I^{-1}(\zeta_{k+1}|\bar{y}_{1:k})$. From \cite{Zuo2011} we have the following result,
\begin{equation}
I(\zeta_{k+1}|\bar{y}_{1:k}) = D_k^{22} - D_k^{21}\left[D_k^{11} + I_A(\zeta_k|\bar{y}_{1:k}) \right]^{-1}D_k^{12},
\end{equation}
where 
\begin{align*}
&D_k^{11} = \mathbb{E}_{f_{k+1}^c}\left[-\triangle_{\zeta_k}^{\zeta_k} \mbox{log } f(\zeta_{k+1}|\zeta_k,\bar{y}_{1:k}) \right], \\ 
&D_k^{12} =   \mathbb{E}_{f_{k+1}^c}\left[-\triangle_{\zeta_k}^{\zeta_{k+1}} \mbox{log } f(\zeta_{k+1}|\zeta_k,\bar{y}_{1:k}) \right]  = (D_k^{21})^T, \\
&D_k^{22} = \mathbb{E}_{f_{k+1}^c}\left[-\triangle_{\zeta_{k+1}}^{\zeta_{k+1}} \mbox{log } f(\zeta_{k+1}|\zeta_k, \bar{y}_{1:k}) f(\bar{y}_{k+1}|\zeta_{k+1})  \right].
\end{align*}
In addition,
\begin{align}
I_A (\zeta_k|\bar{y}_{1:k}) &= E_k^{22} - E_k^{21}\left(E_k^{11}\right)^{-1}E_k^{12},
\end{align}
where 
\begin{align*}
E_k^{11} &= \mathbb{E}_{f(\zeta_{0:k}|\bar{y}_{1:k})}\left[-\triangle_{\zeta_{0:k-1}}^{\zeta_{0:k-1}} \mbox{log } f(\zeta_{0:k}|\bar{y}_{1:k}) \right], \\
E_k^{12} &= \mathbb{E}_{f(\zeta_{0:k}|\bar{y}_{1:k})}\left[-\triangle_{\zeta_{0:k-1}}^{\zeta_{k}} \mbox{log } f(\zeta_{0:k}|\bar{y}_{1:k}) \right] = (E_k^{21})^{T}, \\
E_k^{22} &=  \mathbb{E}_{f(\zeta_{0:k}|\bar{y}_{1:k})}\left[-\triangle_{\zeta_{k}}^{\zeta_{k}} \mbox{log } f(\zeta_{0:k}|\bar{y}_{1:k}) \right].
\end{align*}
We observe that it is still difficult to obtain matrices $E_k^{11},E_k^{12},E_k^{21},E_k^{22}$ so \cite{Zuo2011} introduces the following approximate recursion 
\begin{equation}
I_A (\zeta_k|\bar{y}_{1:k}) \approx S_k^{22} - S_k^{12~T}\left[S_k^{11} + I_A(\zeta_{k-1}|\bar{y}_{1:k-1}) \right]^{-1}S_k^{12},
\end{equation}
where
\begin{align*}
&S_k^{11} = \mathbb{E}_{f(\zeta_{0:k}|\bar{y}_{1:k})}\left[-\triangle_{\zeta_{k-1}}^{\zeta_{k-1}} \mbox{log } f(\zeta_{k}|\zeta_{k-1},\bar{y}_{1:k-1}) \right], \\ 
&S_k^{12} =   \mathbb{E}_{f(\zeta_{0:k}|\bar{y}_{1:k})}\left[-\triangle_{\zeta_{k-1}}^{\zeta_k} \mbox{log } f(\zeta_{k}|\zeta_{k-1},\bar{y}_{1:k-1}) \right], \\
&S_k^{22} = \mathbb{E}_{f(\zeta_{0:k}|\bar{y}_{1:k})}\left[-\triangle_{\zeta_k}^{\zeta_k} \mbox{log } f(\zeta_{k}|\zeta_{k-1}, \bar{y}_{1:k-1}) f(\bar{y}_{k}|\zeta_{k})  \right].
\end{align*}
We observe that in practice it may still be difficult to compute the exact expectations because high dimensional integration is generally involved. Nonetheless, particle filters as described in \cite{Arulampalam2002} can be used to approximate these expectations. Alternative approximations for the  conditional posterior Cramer-Rao lower bound can be found in \cite{Zheng2012}. Unconditional bounds can be found in \cite{Tichavsky1998}.

\subsection{Bounds on Detection}
The algorithm described allows us to compute an approximate lower bound on the mean square error matrix of the attacker's state $\zeta_k$ for a given set of inputs $u_{0:k}^a, \bar s_{1:k}^a$ and observation history $\bar y_{1:k}$. This allows us to obtain a lower bound on the value of $g_k(z_k)$ as follows.
\begin{theorem}
Consider the special case that $\{C_j\}$ is known to the adversary for all $j \in \mathbb{Z}$. Suppose an attacker attempts to estimate $\zeta_{k} = \{\bar{x}_{k},\hat{\bar{x}}_{k},\mathcal{P}_{k}\}$ as in attack strategy 2. Let $\hat{\bar{x}}_{k}^e(\bar{y}_{k})$ be an estimate of $\hat{\bar{x}}_{k}$ as a function of $\bar{y}_{k}$ given $\bar y_{1:k-1}$ and $\hat{e}_{k} = \hat{\bar{x}}_{k}-\hat{\bar{x}}_{k}^e(\bar{y}_{k})$. Suppose a lower bound $Z$ on the error matrix of $\hat{\bar{x}}_{k}$ is obtained so that
\begin{equation}
\mathbb{E}_{f_{k}^c}\left[ \hat{e}_{k}\hat{e}_{k}^T \right]\ge Z_k. \label{assumption}
\end{equation}
Then we have
\begin{equation}
\underset{{\bar{y}}_{k}^a}{\mbox{min}}~ \mathbb{E}_{f^*}\left[g_{k}(z_{k})\right] \ge \mbox{tr}(\mathcal{C}_{k}^T\bar{\mathcal{P}}_{k}^{-1}\mathcal{C}_{k} Z_k),
\end{equation}
where  $f^* = {f(\hat{\bar{x}}_{k}, \bar{y}_{k} |\mathcal{I}_{k-1}^{A \cup P},{u}_{k-1}^a, \bar{s}_{k-1}^a, {u}_{k-1})}$. 
\end{theorem}
\vspace{.2cm}
 \begin{proof}
 First, observe from remark \ref{fk1c}
 \begin{equation}
 f(\zeta_{0:k},\bar{y}_{k}|\mathcal{I}_{k-1}^{A \cup P},{u}_{k-1}^a, \bar{s}_{k-1}^a, {u}_{k-1}) = f_{k}^c.  \label{fstar}
 \end{equation}
 We now have the following.
 \begin{align}
 &\underset{{\bar{y}}_{k}^a}{\mbox{min}}~ \mathbb{E}_{f^*}\left[g_{k}(z_{k})\right] \\
 &=  \underset{{\bar{y}}_{k}^a}{\mbox{min}}~ \mathbb{E}_{f^*}\left[\mbox{tr}\left((\bar{y}_{k}^a-\mathcal{C}_{k}\hat{\bar{x}}_{k})(\bar{y}_{k}^a-\mathcal{C}_{k}\hat{\bar{x}}_{k})^T \bar{\mathcal{P}}_{k}^{-1}\right)\right], \nonumber \\
  &=  \underset{{\bar{y}}_{k}^a}{\mbox{min}}~ \mbox{tr}\left(\mathbb{E}_{f^*}\left[(\bar{y}_{k}^a-\mathcal{C}_{k}\hat{\bar{x}}_{k})(\bar{y}_{k}^a-\mathcal{C}_{k}\hat{\bar{x}}_{k})^T \bar{\mathcal{P}}_{k}^{-1}\right]\right), \nonumber \\
    &=   \mbox{tr}\left(\underset{{\bar{y}}_{k}^a}{\mbox{min}}~\left(\mathbb{E}_{f^*}\left[(\bar{y}_{k}^a-\mathcal{C}_{k}\hat{\bar{x}}_{k})(\bar{y}_{k}^a-\mathcal{C}_{k}\hat{\bar{x}}_{k})^T \right] \right) \bar{\mathcal{P}}_{k}^{-1}\right), \nonumber \\
        &=   \mbox{tr}\left(\underset{\hat{\bar{x}}_{k}^e}{\mbox{min}}~\left(\mathbb{E}_{f^*}\left[(\hat{\bar{x}}_{k}^e-\hat{\bar{x}}_{k})(\hat{\bar{x}}_{k}^e-\hat{\bar{x}}_{k})^T \right] \right) \mathcal{C}_k^T \bar{\mathcal{P}}_{k}^{-1}  \mathcal{C}_k\right), \nonumber \\
           &=   \mbox{tr}\left(\underset{\hat{\bar{x}}_{k}^e}{\mbox{min}}~\left(\mathbb{E}_{f_k^c}\left[(\hat{\bar{x}}_{k}^e-\hat{\bar{x}}_{k})(\hat{\bar{x}}_{k}^e-\hat{\bar{x}}_{k})^T \right] \right) \mathcal{C}_k^T \bar{\mathcal{P}}_{k}^{-1}  \mathcal{C}_k\right), \nonumber \\
                       &\ge \mbox{tr}(\mathcal{C}_{k}^T\bar{\mathcal{P}}_{k}^{-1}\mathcal{C}_{k} Z_k). \nonumber
 \end{align}
 The first two equalities follow from properties of the trace and expectation. The third equality follows from monotonicity properties of the trace function and the fact that $\bar{\mathcal{P}}_k^{-1}$ is constant with respect to $f^*$. The fourth equality is based on the fact that given $\mathcal{C}_k$, a minimizer lies in the range space of $\mathcal{C}_k$. The fifth equality is due to \eqref{fstar}. The final inequality follows from \eqref{assumption}. 
 \end{proof}
\begin{remark}
In general, the adversary's ability to estimate $\{\zeta_k\}$ is dependent on the inputs $\{u_k^a\}, \{\bar s_k^a\}$. For instance, the more the adversary biases the state away from its expected region of operation, the more challenging it is to perform estimation. Thus, if the system operator wishes to analyze how well an adversary can generate stealthy outputs, he must consider a particular sequence of attack inputs $u_k^a, \bar s_k^a$.
\end{remark}
\begin{remark}
In practice, it may be difficult to perform performance analysis when assuming $\mathcal{P}_k$ is an unknown state. However, one can still approximate a lower bound on the error matrix by assuming that the adversary has an oracle which allows him to know $\mathcal{P}_k$, $\mathcal{K}_k$, $I-\mathcal{K}_k\mathcal{C}_k$. 
\end{remark}

%

 \section{Conclusion}
 In this paper, we have considered attacks on control systems where an adversary has access to all channels in a communication network. In order to counter such an adversary, we propose introducing time-varying dynamics into the system which are unknown to the adversary and can in turn be leveraged to detect attacks. Future work will consider sufficient conditions for the design of these matrices to prevent zero-dynamic attacks and the analysis of optimal identification techniques for the adversary.
 
\bibliographystyle{IEEEtran}
\bibliography{IEEEfull,CDC2015_SeanWeerakkody2}

\end{document}